\documentclass[11pt]{article}
\usepackage[margin=1in]{geometry}
\usepackage{amsmath,amssymb,amsthm,mathtools}
\usepackage{thm-restate}
\usepackage{microtype}
\usepackage{xcolor}
\usepackage[colorlinks=true,linkcolor=blue,citecolor=green!45!black,urlcolor=blue]{hyperref}
\usepackage[nameinlink,noabbrev]{cleveref}
\usepackage{enumitem}
\usepackage[ruled,section]{algorithm}
\usepackage{algpseudocode}
\usepackage{booktabs}
\usepackage{array}
\usepackage{placeins}
\usepackage{pifont}
\usepackage{lmodern}

\newcolumntype{P}[1]{>{\raggedright\arraybackslash}p{#1}}

\newtheorem{theorem}{Theorem}[section]
\newtheorem{lemma}[theorem]{Lemma}
\newtheorem{corollary}[theorem]{Corollary}

\newcommand{\restate}[1]{\csname #1\endcsname*}
\theoremstyle{definition}
\newtheorem{remark}[theorem]{Remark}
\numberwithin{equation}{section}

\newcommand{\headingperiod}{}

\newcommand{\mybold}[1]{#1}%

\newcommand{\bx}{\mybold{x}}
\newcommand{\bw}{\mybold{w}}

\newcommand{\eps}{\varepsilon}
\renewcommand{\epsilon}{\varepsilon}

\newcommand{\A}{\mathcal{A}}

\newcommand{\I}{\mathcal{I}}

\newcommand{\M}{\mathcal{M}}

\newcommand{\U}{\mathcal{U}}

\newcommand{\otilde}{\widetilde O}
\newcommand{\thetatilde}{\widetilde\Theta}
\newcommand{\defeq}{\stackrel{\mathrm{{\scriptscriptstyle def}}}{=}}

\DeclareMathOperator{\poly}{poly}
\DeclareMathOperator{\polylog}{polylog}
\DeclareMathOperator{\supp}{supp}

\title{The Power of Local Marginals:\\An $O(\varepsilon^{-1})$-Aspect-Ratio Reduction for Dynamic Weighted Matching}
\author{Jiale Chen\\Stanford University\\\texttt{jialec@stanford.edu}}
\date{}

\begin{document}
\maketitle

\begin{abstract}

We study dynamic maximum weight matching under edge insertions and
deletions in two settings: maintaining a
$(1\pm\varepsilon)$-approximation to the optimum weight, and maintaining an explicit
$(1-\varepsilon)$-approximate matching. Our main result is
a reduction that transforms instances of polynomial aspect
ratio into instances of aspect ratio $O(\varepsilon^{-1})$. The
reduction applies to general graphs in both settings and is compatible with partially dynamic updates.

The reduction is based on a structural property of local marginals.
After grouping edges into weight classes, the global marginal
contribution of one class relative to all lower classes is approximated
by its marginal contribution within a local weight window of aspect
ratio $O(\varepsilon^{-1})$. Summing these local marginals yields a value
composition lemma that approximates the optimum weight in the entire graph with approximate optimum weights of the local
windows. This improves the value reduction of Gupta and Peng
(FOCS 2013), whose local aspect ratio is
$\varepsilon^{-\Theta(\varepsilon^{-1})}$. The same structural property
yields an improved matching composition lemma for explicit matchings,
reducing the local aspect ratio of Bernstein--Chen--Dudeja--Langley--Sidford--Tu (SODA 2025) from
$O(\varepsilon^{-2})$ to $O(\varepsilon^{-1})$.

\end{abstract}
\tableofcontents
\newpage
\section{Introduction\headingperiod}

In an $n$-vertex weighted graph $G=(V,E,w)$ with edge weights
$w:E\to\mathbb R_{>0}$, a matching $M$ is a set of vertex-disjoint edges
with weight $w(M)=\sum_{e\in M}w(e)$.\footnote{Throughout this paper, we assume
that the edge weights lie in a fixed interval of aspect ratio at most
$\poly(n)$.}
A \emph{maximum-weight matching}
(MWM) is a matching of maximum weight, and we denote its weight by
$\mu_w(G)$.  When every edge has unit weight, matching weight equals
cardinality. The corresponding problem is \emph{maximum-cardinality
matching} (MCM), and we denote the maximum matching cardinality by $\mu(G)$.

We study the dynamic matching problem, in which a graph
$G$ undergoes edge updates
(insertions and deletions) and the algorithm maintains a
large matching after every update.
The time taken for the algorithm to maintain the matching
between edge updates is the \emph{update time}.
For exact maximum matching, conditional lower bounds rule out
$O(n^{1-\delta})$ amortized update time for every constant $\delta>0$, even
for bipartite MCM under only insertions or only deletions~\cite{Dahlgaard16}.

The barriers motivate approximation.  The ultimate goal is a dynamic
algorithm that maintains a
$(1-\eps)$-approximate matching $M$,
in $\poly(\log n,\eps^{-1})$ update
time~\cite{AbboudRW17,Liu24}. 
A $(1-\eps)$-approximate
MWM, or $(1-\eps)$-MWM, has weight
$w(M)\geq(1-\eps)\mu_w(G)$. The correspondence in a unit-weight graph
is called a $(1-\eps)$-MCM, with cardinality
$|M|\geq(1-\eps)\mu(G)$.
In a partially dynamic graph, where updates consist only of edge insertions
(incremental) or only of edge deletions (decremental),
$\poly(\log n,\eps^{-1})$-update-time results are
known~\cite{Gupta14,BlikstadK23,ChenST25,BernsteinCDLST25}.

The fully
dynamic setting remains significantly different even for $(1-\eps)$-MCM:
no polynomial improvement over
the linear update-time barrier is known for maintaining
an explicit matching on dense graphs,
and the best bounds remain
$n^{1-o(1)}$~\cite{GuptaP13,AssadiBKL23,Liu24,MitrovicS25}.
Recent work parameterizes the update time by the density of Ordered
Ruzsa--Szemer\'edi (ORS) graphs, which pack a sequence of linear-size
matchings that are
induced relative to later ones~\cite{BehnezhadG24,AssadiKK25,MitrovicS25}.
Later it is shown by \cite{Pratt26} that ORS density is essentially equivalent to classical
RS density, which is a longstanding open
problem in extremal combinatorics, and the current bounds do not yet yield a
polynomial update-time improvement.

Alternatively, for the value version of the problem of
maintaining a $(1\pm\eps)$-approximation $\nu$ to $\mu(G)$,
i.e., $|\nu-\mu(G)|\leq \eps\mu(G)$, \cite{BhattacharyaKS23Size}
gives a randomized fully dynamic algorithm
with $m^{1/2-\Omega_\eps(1)}$ worst-case update time with high probability
against an adaptive adversary, giving the first polynomial improvement
over the $O(n)$ barrier on dense graphs.

For weighted graphs, a central line of work develops black-box
reductions from dynamic MWM to dynamic MCM. 
The first general
reduction~\cite{StubbsW17} loses roughly a factor of two in the
approximation ratio.  For bipartite graphs,
\cite{BernsteinDL21} subsequently obtained the first low-loss
reduction,
transforming a dynamic $(1-\eps)$-MCM algorithm into a dynamic
$(1-\eps)$-MWM algorithm,
but with $\eps^{-\Theta(\eps^{-1})}\log n$ overhead.
\cite{BernsteinCDLST25} later reduced the overhead to
$\poly(\eps^{-1})$ in bipartite graphs.  More recently, \cite{BernsteinC26} developed the first
low-loss reduction for fully dynamic general graphs, incurring
$\poly(\log n,\eps^{-1})$ overhead.

The two low-loss reductions with $\poly(\log n,\eps^{-1})$ overhead
share the same elegant weight reduction.
The matching
composition lemma of~\cite{BernsteinCDLST25} decomposes an instance
with $\poly(n)$ \emph{aspect ratio}---the ratio between its largest and
smallest edge weights---into subgraphs of
$O(\eps^{-2})$ aspect ratio where each
subgraph contains every edge with weight in a given window.
It then shows that the union of the local
$(1-\eps)$-MWMs in each window
contains a global $(1-\eps)$-MWM, thereby reducing
$\poly(n)$ aspect ratio to $O(\eps^{-2})$ with only
$\poly(\eps^{-1})$ additive overhead. For bipartite graphs,
\cite{BernsteinCDLST25} combines this lemma with the dynamic
graph unfolding framework of \cite{BernsteinDL21}.  For general graphs,
\cite{BernsteinC26} combines the weight range reduction with a
primal-dual algorithm based on induced-matching queries simulated
by fully dynamic bipartite $(1-\eps)$-MCM algorithms.

Both
\cite{BernsteinCDLST25,BernsteinC26} require the MCM
algorithm to return an explicit matching. Consequently, neither reduction
applies to the value-only breakthrough above. For bipartite graphs, the earlier 
black-box reduction of \cite{BernsteinDL21} is essentially value-compatible. However, it uses the dynamic weight reduction framework of \cite{GuptaP13} that incurs
$\eps^{-\Theta(\eps^{-1})}\log n$ overhead
(see \cite{BernsteinCDLST25} for a detailed discussion
on the limitation of \cite{GuptaP13}).

This leaves our central question, already open for bipartite graphs:
Can a dynamic algorithm that only estimates the MCM cardinality
be converted into a dynamic algorithm that estimates
the MWM weight, with only $\poly(\log n,\eps^{-1})$ overhead?

We answer this question affirmatively for bipartite graphs.
We give the first low-loss
value-compatible weight reduction 
on general graphs
that decreases the aspect ratio from $\poly(n)$ to
$O(\eps^{-1})$.
Combined with
graph unfolding \cite{KaoLST01, BernsteinDL21}, which is specific to bipartite graphs,
it transforms a dynamic bipartite MCM cardinality algorithm
into an MWM weight algorithm with only
$\otilde(\eps^{-2})$ overhead.\footnote{In this paper, $\otilde(\cdot)$ hides only factors polynomial in
$\log(\eps^{-1})$. All logarithms have base 2.}
The weight reduction is based on a simple 
observation on the marginal contribution of
local weight windows.
The same observation also
leads to an improved weight reduction for explicit matchings. 
We now explain the observation below.

\paragraph{The Power of Local Marginals.}
Denote by $R$ the aspect ratio of the graph. Rescale the edge weights to $[1,R]$ and
partition the edges into classes $E_i=\{e:2^i\le w(e)<2^{i+1}\}$.
Let $G_{[a,b]}$ denote the graph induced by
$E_a\cup\cdots\cup E_b$.
Let $G_i=G_{[i,i]}$ and $G_{\le i}=G_{[0,i]}$.

The matching composition lemma of 
\cite{BernsteinCDLST25} states a structural 
property of explicit local matchings: the 
union of approximate matchings computed on 
overlapping weight windows contains a
near-optimal matching for the whole graph. 

In the value-only setting
where explicit matchings 
are unavailable, we instead 
seek a quantity that can be composed using 
only local optimum values.
Define the marginal contribution of $E_i$ w.r.t.\ $h$ lower classes as
$\Delta_i^{(h)}\defeq \mu_w(G_{[i-h,i]})-\mu_w(G_{[i-h,i-1]})$.
Taking $h=i$ exposes the entire
prefix and $\Delta_i^{(i)}=\mu_w(G_{\leq i})-\mu_w(G_{\leq i-1})$
is the global contribution of $E_i$.
The key observation shows that if we take
$h=\lceil\log(\eps^{-1})\rceil+\Theta(1)$,
corresponding to aspect ratio $O(\eps^{-1})$,
$\Delta_i^{(h)}$ gives a sufficiently accurate approximation to $\Delta_i^{(i)}$.

\begin{restatable}{lemma}{cor:marginal-sandwich}
\label{cor:marginal-sandwich}
For any class $E_i$ and any integer $h\ge0$,
$\lvert\Delta_i^{(h)}-\Delta_i^{(i)}\rvert\leq 2^{1-h}\mu_w(G_i)$.
\end{restatable}

The local contribution can also be realized by explicit matchings,
which naturally leads to an
improved matching composition lemma.
Let $M_i$ be an approximate MWM of the 
local window, and let $P_{i-1}$ be any
matching accumulated from the lower scales. The increase in matching weight
obtained from $P_{i-1}\cup M_i$
approximately realizes the global contribution $\Delta_i^{(i)}$.
\begin{restatable}{lemma}{lem:prefix-extension}
\label{lem:prefix-extension}
For every $i$, $h\geq 1$, and $0\leq\rho\leq 1$, let $M_i$ be a
$(1-\rho)$-MWM of $G_{[i-h,i]}$ and let
$P_{i-1}\subseteq G_{\leq i-1}$ be any matching.  Then
\[
 \mu_w(P_{i-1}\cup M_i)-w(P_{i-1})
 \geq \Delta_i^{(i)}-\rho\mu_w(G_{[i-h,i]})-2^{2-h}\mu_w(G_i).
\]
\end{restatable}

Both \Cref{cor:marginal-sandwich,lem:prefix-extension} are short consequences of a
single local exchange lemma (\Cref{lem:marginal-exchange}) that we show at the beginning of \Cref{sec:composition}.

\paragraph{Value composition.}
Since the global contributions telescope, i.e.,
$\sum_i \Delta_i^{(i)}=\mu_w(G)$, \Cref{cor:marginal-sandwich} implies
that the sum of the local contributions $\Delta_i^{(h)}$ approximates the global MWM
weight.  This remains true when each local optimum is replaced by approximations.  Consequently, an
$O(\eps^{-1})$-aspect-ratio MWM weight algorithm yields a
$\poly(n)$-aspect-ratio algorithm with $\otilde(1)$ overhead.
Combining this reduction with graph
unfolding~\cite{BernsteinDL21,KaoLST01} converts an MCM cardinality algorithm
into an MWM weight algorithm for bipartite graphs
with $\otilde(\eps^{-2})$ overhead
(see \Cref{thm:value-w2u}).

\paragraph{Matching composition.}
\Cref{lem:prefix-extension} gives the matching counterpart.
We process the windows from low weights to high weights.  After adding each
local matching $M_i$, we maintain a large matching from the union of $M_i$ and
the current partial matching.  Summing the gains over all steps shows that
local matchings from $O(\eps^{-1})$-aspect-ratio windows compose into a
$(1-\eps)$-MWM.

For the dynamic implementation,
immediately propagating every change in a local matching can
create a cascade through all later windows.
We instead rebuild a partial matching only when necessary to preserve the
approximation guarantee, thereby preventing such cascades.
An $O(\eps^{-1})$-aspect-ratio algorithm can be transformed
into a $\poly(n)$-aspect-ratio algorithm with $\otilde(\eps^{-3})$
overhead (see \Cref{thm:dynamic}).
Compared with~\cite{BernsteinCDLST25}, this improves the local aspect ratio
from $\otilde(\eps^{-2})$ to $O(\eps^{-1})$ and the additive update overhead
from $\otilde(\eps^{-6})$ to $\otilde(\eps^{-3})$.

\paragraph{Tightness of the local aspect ratio.}
The $O(\eps^{-1})$ local aspect ratio is asymptotically tight for the two
forms of composition above.  Reductions using only the MWM value of each
weight window require a local aspect ratio $\Omega(\eps^{-1})$.  The same
lower bound holds for matching composition. However, the value lower bound
does not cover reductions using
additional information from the windows, and the matching lower bound
assumes adversarial
choices of the local MWMs.  See \Cref{sec:lower}.

\paragraph{Applications.}
Our dynamic weight reduction for explicit matching can replace that of
\cite{BernsteinCDLST25} in all of its dynamic applications.  Here we highlight two applications.
For the value version, applying our reduction to
\cite{BhattacharyaKS23Size} yields, for
every fixed $\eps>0$, a randomized fully dynamic algorithm that maintains a
$(1\pm\eps)$-approximation to the MWM weight of a $\poly(n)$-aspect-ratio
bipartite graph in $m^{1/2-\Omega_\eps(1)}$ worst-case update time, with high
probability against an adaptive adversary.
This transfers their polynomial
improvement from MCM cardinality to MWM weight in bipartite graphs.
For fixed $\eps$,
combining~\cite{GuptaP13,BernsteinDL21,BhattacharyaKS23Size}
gives $O(m^{1/2-\Omega_\eps(1)}\log n)$ update time.
Our reduction removes the $\log n$ overhead and improves the
dependence on $1/\eps$ in the reduction overhead from exponential
to polynomial.

For explicit matching, a new application is to
low-arboricity graphs.  We design a $(1-\eps)$-MWM sparsifier with
$O(\alpha\eps^{-1})$ degree on a graph of arboricity $\alpha$.
It is a natural extension of the $(1-\eps)$-MCM sparsifier in
low-arboricity graphs (see \cite{PelegS16,Solomon18}).
This gives a fully dynamic algorithm for
explicitly maintaining a $(1-\eps)$-MWM
in amortized update time $\otilde(\alpha\eps^{-4})$.
On bipartite graphs, an
$\alpha\poly(\eps^{-1})$ amortized update bound already follows by
combining~\cite{PelegS16} with the reduction
of~\cite{BernsteinCDLST25} because
graph unfolding preserves arboricity.
On general graphs, to our knowledge, this is the first result
depending only on $\alpha$ and $\eps$.

Beyond these applications, our reduction improves the
$\eps$-dependence of several explicit-matching bounds from
\cite{BernsteinCDLST25}.  Specifically, it gives deterministic
$(1-\eps)$-MWM algorithms
with total time $\otilde(n\eps^{-8}\log n+m\eps^{-7})$ in incremental bipartite
graphs, amortized update time $\otilde(\Delta\eps^{-3})$ in fully dynamic
maximum-degree-$\Delta$ graphs, and amortized update time
$\otilde(\sqrt m\,\eps^{-3})$ in general fully dynamic graphs. It also improves the dynamic rounding algorithm in weighted general graphs to $\otilde(\eps^{-6}\polylog n)$ amortized update time.

Our matching composition lemma also applies to the
semi-streaming and linear-memory MPC reductions
of~\cite{BernsteinCDLST25}. Using \Cref{thm:composition} in place
of their composition lemma reduces the required local aspect ratio
to $O(\eps^{-1})$, with local approximation error
$\Theta(\eps/\log(1/\eps))$. This also reduces the size of the
unfolded instances in their bipartite weight-to-cardinality reductions.
We refer readers to~\cite[Sections~6.2 and~6.3]{BernsteinCDLST25}
for the implementations in these models. This paper focuses on the
dynamic setting.

\section{Value and Matching Composition\headingperiod}\label{sec:composition}
In this section, we first prove the local exchange lemma that enables the proof of
\Cref{cor:marginal-sandwich,lem:prefix-extension}, and then obtain the value
composition and matching composition with them.

Recall that we rescale the edge weights to $[1,R]$ and
partition the edges into classes $E_i=\{e:2^i\le w(e)<2^{i+1}\}$.
Let $G_{[a,b]}$ denote the graph induced by
$E_a\cup\cdots\cup E_b$ with $G_{[a,b]}=\emptyset$ if $a>b$.
Let $G_i=G_{[i,i]}$ and $G_{\le i}=G_{[0,i]}$.

The local exchange lemma will show that a matching
in the prefix $G_{\le i-1}$ and a local matching in $G_{[i-h,i]}$
can be repartitioned into
matchings in $G_{\le i}$ and $G_{[i-h,i-1]}$ with only a
small loss in total weight. The reverse direction also
holds.

\begin{lemma}[Local Exchange]\label{lem:marginal-exchange}
Fix an integer $i\ge0$ and $h\ge1$.
\begin{enumerate}[label=\textup{(\roman*)},leftmargin=2.1em]
\item For matchings $A\subseteq G_{\le i-1}$ and $B\subseteq G_{[i-h,i]}$,
there are matchings
$C\subseteq G_{\leq i}$, $D\subseteq G_{[i-h,i-1]}$ such that
$C, D\subseteq A\cup B$ and
\begin{equation}
 w(C)+w(D)
 \ge w(A)+w(B)-2^{1-h}\mu_w(G_i).
\label{eq:exchange-forward}
\end{equation}
\item For matchings $A\subseteq G_{\le i}$ and $B\subseteq G_{[i-h,i-1]}$,
there are matchings $C\subseteq G_{[i-h,i]}$ and
$D\subseteq G_{\le i-1}$ such that
$C, D\subseteq A\cup B$ and
\begin{equation}
 w(C)+w(D)
 \ge w(A)+w(B)-2^{1-h}\mu_w(G_i).
\label{eq:exchange-backward}
\end{equation}
\end{enumerate}
\end{lemma}

\begin{proof}
We first prove (i). Put every common edge of $A$ and
$B$ into both output matchings.  The remaining edges form vertex-disjoint
alternating paths and even cycles.

On each alternating component, start at every $B$-edge in $E_i$ and walk in
both directions.  In each direction, if it exists, delete the first edge of
weight below $2^{i-h}$.  After all cuts, no remaining component contains both
an edge of $E_i$ and an edge below $2^{i-h}$.  Every $B$-edge in $E_i$
causes at most two
deletions, and every edge it deletes has weight less than $2^{-h}$ times its
own weight.  The total deleted weight is therefore at most
$2^{1-h}w(B\cap E_i)\leq 2^{1-h}\mu_w(G_i)$.

Assign the remaining components as follows.  On a component containing an edge
in $E_i$, put the $B$-edges in $C$ and the $A$-edges in $D$.  On a component
containing an edge below $2^{i-h}$, put the $A$-edges in $C$ and the
$B$-edges in $D$.  Assign either way on every remaining component.  The
alternation makes both outputs matchings.  Moreover,
$C\subseteq A\cup B\subseteq G_{\leq i}$ and $D$ has neither
an edge in $E_i$ nor an edge below $2^{i-h}$, so $D\subseteq
G_{[i-h,i-1]}$.  This proves \eqref{eq:exchange-forward}.

For (ii), apply the identical cuts corresponding to
$A$-edges in $E_i$, charging the weight loss to
$A\cap E_i$.  On a component containing edges in $E_i$, send the $A$-edges to
$C$ and the $B$-edges to $D$; on every other component, do the reverse.
Then $C\subseteq G_{[i-h,i]}$ and $D\subseteq G_{\le i-1}$, which proves
\eqref{eq:exchange-backward}.
\end{proof}

\subsection{Value Composition\headingperiod}

To see how local exchange is useful to value composition,
we prove \Cref{cor:marginal-sandwich}, which shows that the local marginal contribution
$\Delta_i^{(h)}=\mu_w(G_{[i-h,i]})-\mu_w(G_{[i-h,i-1]})$ approximates the
global marginal contribution $\Delta_i^{(i)}=\mu_w(G_{\le i})-\mu_w(G_{\le i-1})$.

\restate{cor:marginal-sandwich}
\begin{proof}[Proof of \Cref{cor:marginal-sandwich}]
For $h=0$, we have $\Delta_i^{(h)}=\mu_w(G_i)$ and
$0\leq \Delta_i^{(i)}\leq\mu_w(G_i)$, so the claimed bound holds. Thus assume $h\geq1$.
Applying \Cref{lem:marginal-exchange}(i) to MWMs $A$ of $G_{\le i-1}$ and
$B$ of $G_{[i-h,i]}$, we have $C\subseteq G_{\le i}$,
$D\subseteq G_{[i-h,i-1]}$ and $\mu_w(G_{\le i})+\mu_w(G_{[i-h,i-1]})
\ge \mu_w(G_{\le i-1})+\mu_w(G_{[i-h,i]}) -2^{1-h}\mu_w(G_i)$.
Similarly, applying \Cref{lem:marginal-exchange}(ii) to MWMs $A$ of $G_{\le i}$ and
$B$ of $G_{[i-h,i-1]}$, we have
$\mu_w(G_{[i-h,i]})+\mu_w(G_{\le i-1})
\ge \mu_w(G_{\le i})+\mu_w(G_{[i-h,i-1]}) -2^{1-h}\mu_w(G_i)$.
\end{proof}

The sum of local marginal contributions approximates the MWM weight of $G$. The
error term is an accumulation of local MWM weights which could be bounded
by the following lemma.

\begin{lemma}\label{lem:interval-sum}
For every integer $k\ge0$,
\begin{equation}
 \sum_{i\ge0}\mu_w(G_{[i-k,i]})\le5(k+1)\mu_w(G).
\end{equation}
\end{lemma}

\begin{proof}
For every $i$, choose an MWM $A_i$ of $G_{[i-k,i]}$ and partition the
indices by their residue modulo $k+1$.  Fix one residue and process its
nonempty $A_i$ in decreasing order of $i$.  Greedily add every edge that
does not conflict with an edge already retained, and charge a discarded edge
to one retained conflicting edge from a higher window.

Suppose a retained edge $e$ came from $A_i$.  The $t$-th lower window in the
same residue class ends at $i-t(k+1)$.  It contributes at most two edges that
conflict with $e$, and each such edge $f$ satisfies
\[
 \frac{w(f)}{w(e)}
 <\frac{2^{i-t(k+1)+1}}{2^{i-k}}
 =2^{-(t-1)(k+1)}.
\]
Consequently, the weight of $e$ together with everything charged to it is at
most
\[
 w(e)\left(1+2\sum_{t\ge1}2^{-(t-1)(k+1)}\right)\le5w(e).
\]
The retained edges form a matching in $G$.  Thus each residue class
contributes at most $5\mu_w(G)$, and there are $k+1$ residues.
\end{proof}

Now we give the value composition lemma.

\begin{lemma}[Value Composition]\label{thm:scalar-composition}
Let $h\ge1$ and $0\le\rho\le1$.  Suppose, for every $0\le i\le L\defeq \lceil\log R\rceil$, we
have $(1\pm\rho)$-approximate values of $\mu_w(G_{[i-h,i]})$ and
$\mu_w(G_{[i-h,i-1]})$, denoted by
$\widehat\mu_w(G_{[i-h,i]})$ and
$\widehat\mu_w(G_{[i-h,i-1]})$, respectively,
and define
\begin{equation}
 \widehat V_h\defeq\sum_{i=0}^{L}
 \left(\widehat\mu_w(G_{[i-h,i]})
       -\widehat\mu_w(G_{[i-h,i-1]})\right).
\end{equation}
Then
\begin{equation}
 \left|\widehat V_h-\mu_w(G)\right|
 \le\left(10\cdot2^{-h}+5(2h+1)\rho\right)\mu_w(G).
\end{equation}
In particular, for $0<\eps\le1/2$, setting $h=\left\lceil\log(20/\eps)\right\rceil$
and $\rho=\frac{\eps}{10(2h+1)}$
makes $\widehat V_h$ a $(1\pm\eps)$-approximation of $\mu_w(G)$, and any
nonempty $G_{[i-h,i]}$ or $G_{[i-h,i-1]}$ has aspect ratio at most
$80/\eps$.
\end{lemma}

\begin{proof}
Let $V_h$ denote the exact version of $\widehat V_h$, i.e., $
 \sum_{i=0}^{L}\left(\mu_w(G_{[i-h,i]})-\mu_w(G_{[i-h,i-1]})\right)$.
Summing the bound in \Cref{cor:marginal-sandwich} over $i$, using telescoping for
the prefix differences, and applying \Cref{lem:interval-sum} with $k=0$
gives
\[
 |V_h-\mu_w(G)|
 \le 2^{1-h}\sum_i\mu_w(G_i)
 \le 10\cdot2^{-h}\mu_w(G).
\]
Again, applying \Cref{lem:interval-sum} with $k=h$ and $k=h-1$ gives
\[
 |\widehat V_h-V_h|
 \le\rho\sum_i\left(
 \mu_w(G_{[i-h,i]})+\mu_w(G_{[i-h,i-1]})\right)
 \le5(2h+1)\rho\mu_w(G).
\]
Combining the two bounds completes the proof.
\end{proof}

\subsection{Matching Composition\headingperiod}
For matching composition, we first prove \Cref{lem:prefix-extension}
which shows how to analyze the union of
local $(1-\eps)$-MWMs by adding them one by one.

\restate{lem:prefix-extension}
\begin{proof}[Proof of \Cref{lem:prefix-extension}]
Applying \Cref{lem:marginal-exchange}(i) to $P_{i-1}$ and $M_i$ gives
$C\subseteq P_{i-1}\cup M_i$, 
$D\subseteq G_{[i-h,i-1]}$. Thus
\[
\begin{aligned}
 \mu_w(P_{i-1}\cup M_i)-w(P_{i-1})
 &\ge w(M_i)-\mu_w(G_{[i-h,i-1]})
       -2^{1-h}\mu_w(G_i)\\
 &\ge \mu_w(G_{[i-h,i]})-\mu_w(G_{[i-h,i-1]})
       -\rho\mu_w(G_{[i-h,i]})-2^{1-h}\mu_w(G_i).
\end{aligned}
\]
\Cref{cor:marginal-sandwich} gives
\[
 \mu_w(G_{[i-h,i]})-\mu_w(G_{[i-h,i-1]})
 \ge \mu_w(G_{\le i})-\mu_w(G_{\le i-1})-2^{1-h}\mu_w(G_i).
\]
Combining the two inequalities proves the claim.
\end{proof}

Now if we take each $P_i$ as the MWM on $P_{i-1}\cup M_i$,
summing \Cref{lem:prefix-extension} over $i$ gives the
desired property of the union of the local $(1-\eps)$-MWMs.

\begin{lemma}[Matching Composition]\label{thm:composition}
Let $h\ge1$ and $0\le\rho\le1$.  Suppose, for every
$0\le i\le L\defeq\lceil\log R\rceil$, we have a $(1-\rho)$-MWM
$M_i$ of $G_{[i-h,i]}$.  Then
\[
 \mu_w\!\left(\bigcup_{i=0}^L M_i\right)
 \ge\left(1-20\cdot2^{-h}-5(h+1)\rho\right)\mu_w(G).
\]
In particular, for $0<\eps\le1/2$, setting
$h=\lceil\log(40/\eps)\rceil$ and
$\rho=\eps/[10(h+1)]$ gives
\[
 \mu_w\!\left(\bigcup_{i=0}^LM_i\right)
 \ge(1-\eps)\mu_w(G),
\]
and every nonempty $G_{[i-h,i]}$ has aspect ratio less than $160/\eps$.
\end{lemma}

\begin{proof}
Define $P_i$ as the MWM on $P_{i-1}\cup M_i$ with $P_{-1}=\emptyset$,
thus $P_i\subseteq \bigcup_{j=0}^i M_j$.
By \Cref{lem:prefix-extension},
\begin{equation}\label{eq:prefix-two-degree}
w(P_i)-w(P_{i-1})\ge\mu_w(G_{\le i})-\mu_w(G_{\le i-1})
-\rho\mu_w(G_{[i-h,i]})-2^{2-h}\mu_w(G_i).
\end{equation}
Summing over $i$ and applying \Cref{lem:interval-sum} yields
\[
 w(P_L)
 \ge\left(1-5(h+1)\rho-20\cdot2^{-h}\right)\mu_w(G),
\]
which finishes the proof.
\end{proof}

\subsection{Tightness of the Local Aspect Ratio\headingperiod}\label{sec:lower}

The value result has the best possible dependence on $\eps$ for reductions
that see only MWM weights of subgraphs $G[\ell_i,r_i]$ that contain all
edges with weights in $[\ell_i,r_i]\subseteq\mathbb R_{>0}$ and $r_i/\ell_i\le
O(1/\eps)$.

\begin{lemma}
\label{prop:value-width-lower-bound}
For any $0<\eps<1/2$ and
$1\le R<(\eps^{-1}-1)/4$, there are two graphs $G$ and $G'$ such that, for
every family $\mathcal I=\{[\ell_i,r_i]\}_i$ satisfying
$r_i/\ell_i\le R$ for every $i$,
$\mu_w(G[\ell_i,r_i])=\mu_w(G'[\ell_i,r_i])$ for every $i$,
but no value is a $(1\pm\eps)$-approximation of both $\mu_w(G)$ and
$\mu_w(G')$.
\end{lemma}

\begin{proof}
Use two edges with weight $1$ and $2R$. In the first graph, the two edges are disjoint;
in the second they share an endpoint.  Any interval $[\ell_i,r_i]$
containing both weights has $r_i/\ell_i\ge2R>R$, so all admissible interval
values agree. The global optima are $2R+1$ and
$2R$.  If one answer were a $(1\pm\eps)$-estimate for both, then
$(1-\eps)(2R+1)\le(1+\eps)2R$, which requires
$R\ge(\eps^{-1}-1)/4$.
\end{proof}

A similar statement applies to
matching composition if local optima may be chosen adversarially.

\begin{lemma}
\label{lem:width-lower-bound}
For any $0<\eps<1/2$ and $1\le R<\eps^{-1}-1$, there is a graph $G$ such
that, for every family $\mathcal I=\{[\ell_i,r_i]\}_i$ satisfying
$r_i/\ell_i\le R$ for every $i$, one can choose an MWM $M_i$ of
$G[\ell_i,r_i]$
for each $i$ so that
\[
 \mu_w\!\left(\bigcup_i M_i\right)
 <(1-\eps)\mu_w(G).
\]
\end{lemma}

\begin{proof}
Choose $r$ with $R<r<\eps^{-1}-1$, and let $G$ be the path $a-b-c-d$ with
edge weights $w(ab)=w(bc)=1$ and $w(cd)=r$.  Any interval $[\ell_i,r_i]$
containing both weights has $r_i/\ell_i\ge r>R$ and is therefore not
admissible.  Whenever an interval contains the unit edges, return $\{bc\}$;
whenever it contains the heavy edge, return $\{cd\}$.  The union has optimum
at most $r$, while $\{ab,cd\}$ has weight $1+r$.  Since
$r/(1+r)<1-\eps$, the claim follows.
\end{proof}

\section{Dynamic Weight Reduction for MWM Weight\headingperiod}\label{sec:value-reduction}

The value composition lemma gives a direct dynamic reduction: run the
underlying algorithm on $O(1/\eps)$-aspect-ratio windows and combine the weight
estimates.  Since every window is a subgraph of the input, the reduction applies
within any graph class closed under taking subgraphs, including bipartite and bounded-arboricity graphs.

\begin{theorem}
\label{thm:value-reduction}
Let $\mathcal G$ be a graph family closed under taking subgraphs.
Suppose that, for every $0<\delta\le1/2$, there is a dynamic algorithm
that, on a graph $F\in\mathcal G$ with at most $n$ vertices and $m$ edges, and
aspect ratio at most $R$, maintains a $(1\pm\delta)$-approximation of $\mu_w(F)$
with initialization time $\I(n,m,R,\delta)$ and update time $\U(n,m,R,\delta)$.

Then for every $0<\eps\le1/2$, there is a dynamic algorithm
on a graph
$F\in\mathcal G$ with at most $n$ vertices and $m$ edges, and $\poly(n)$ aspect ratio,
that
maintains a $(1\pm\eps)$-approximation of $\mu_w(F)$ with initialization time\footnote{Throughout this paper, we assume that the initialization time bound
of a dynamic matching algorithm is superadditive in the number of edges, with all other parameters fixed.}
\[\otilde\left(
\I(n,\otilde(m),O(\eps^{-1}),\thetatilde(\eps))+m\right)\]
and update time
\[\otilde\left(\U\!\left(n,m,O(\eps^{-1}),\thetatilde(\eps)\right)\right).\]
The reduction preserves worst-case update time,
determinism, and incremental or decremental updates.
\end{theorem}

\begin{proof}
Set
$h=\lceil\log(20/\eps)\rceil$,
$\rho=\eps/[10(2h+1)]$.
Initialize one copy of the underlying algorithm
on every $G_{[i-h,i]}$ and $G_{[i-h,i-1]}$, and maintain
$\widehat V_h
 =\sum_{i=0}^{L}
   \left(\widehat\mu_w(G_{[i-h,i]})
         -\widehat\mu_w(G_{[i-h,i-1]})\right)$.
An edge in $E_j$ belongs to
$G_{[i-h,i]}$ for $j\le i\le j+h$ and to $G_{[i-h,i-1]}$ for
$j+1\le i\le j+h$.  It therefore affects at most $2h+1$ local instances.
Before and after updating each affected instance, adjust its term in
$\widehat V_h$.  By
\Cref{thm:scalar-composition}, they give a
$(1\pm\eps)$-approximation, and every local window has aspect ratio $O(\eps^{-1})$.

Every local instance is a fixed weight-induced subgraph.  Thus an insertion
in the original graph causes only insertions in $\otilde(1)$ local instances, and a
deletion causes only $\otilde(1)$ deletions.  The reduction introduces neither
randomization nor amortization. The initialization cost is the sum of the initialization
cost of each local window, where any edge contributes to $\otilde(1)$ windows.
\end{proof}

\subsection{Dynamic Weight-to-Cardinality Reduction in Bipartite Graphs\headingperiod}
\label{sec:weight-to-cardinality}

\cite{BernsteinCDLST25,BernsteinDL21} reduces bipartite MWM to MCM
using the graph unfolding structure of \cite{KaoLST01}.
Roughly, unfolding replaces each vertex by
copies and each integer-weight edge by a structured collection of 
unit-weight edges, so that the MCM cardinality in the unfolded graph equals the
MWM weight of the original graph.

\begin{lemma}[\cite{KaoLST01}]\label{lem:unfolding}
For an $n$-node $m$-edge graph $G$ with integer edge weights between $1$ and $R$,
the unfolded graph $\phi(G)$ is an unweighted graph defined as follows. For each 
vertex $u$ in $G$, create $R$ copies of $u$ named $\{u^1,u^2,\dots,u^R\}$ in
$\phi(G)$. For each edge $e$ between $u$ and $v$ in $G$, create edges between $u^
{i}$ and $v^{w(e)-i+1}$ for all $i\in [w(e)]$ in $\phi(G)$. Then $\phi(G)$ has 
at most $Rn$ vertices and $Rm$ edges. Additionally, if $G$ is bipartite, $\mu(\phi(G))=\mu_w(G)$.
\end{lemma}

We combine the graph unfolding structure with \Cref{thm:value-reduction}
to get the following weight-to-cardinality reduction in bipartite graphs.

\begin{theorem}
\label{thm:value-w2u}
Suppose that, for every $0<\delta\le1/2$, there is a dynamic
algorithm that, on a unit-weight bipartite graph $F$ with at most $n$
vertices and $m$ edges, maintains a $(1\pm\delta)$-approximation of $\mu(F)$
in initialization time $\I(n,m,\delta)$ and update time $\U(n,m,\delta)$.

Then for every $0<\eps\le1/2$, there is a dynamic algorithm on a bipartite graph
$F$ with at most
$n$ vertices and $m$ edges, and $\poly(n)$ aspect ratio, that maintains a
$(1\pm\eps)$-approximation of $\mu_w(F)$ in initialization time
\[\otilde\left(\eps^{-2}\I(\otilde(n\eps^{-2}),\otilde(m\eps^{-2}),\thetatilde(\eps))\right)\]
and update time
\[\otilde\left(\eps^{-2}\U(\otilde(n\eps^{-2}),\otilde(m\eps^{-2}),\thetatilde(\eps))\right).\]
The reduction preserves worst-case update time,
determinism, and incremental or decremental updates.
\end{theorem}

\begin{proof}
Set
$h=\lceil\log(20/\eps)\rceil$,
$\rho=\eps/[10(2h+1)]$,
$R=80/\eps$, and
$\tau=\rho/3$.  Let
$K=\lceil Ch/\eps^2\rceil$ for a sufficiently large universal constant
$C$.
Consider one window $F$, and let $\lambda$ be the fixed lower endpoint of
its weight interval, i.e., every edge of $F$
has weight in $[\lambda,R\lambda)$.  Replace each weight by the
positive integer
$\overline w(e)=\left\lfloor\frac{w(e)}{\tau\lambda}\right\rfloor$.
Denote by $\overline F$ the graph after processing.
Since for every matching $M$,
$(1-\tau)w(M)\le\tau\lambda\,\overline w(M)\le w(M)$,
applying it to MWMs of $F$ and $\overline F$ gives
$(1-\tau)\mu_w(F)\le\tau\lambda\mu_w(\overline F)\le\mu_w(F)$.

We run the unweighted algorithm to maintain a $(1\pm\rho/3)$-approximation
$\widehat\mu$ of the MCM cardinality
on $\phi(\overline F)$ defined in \Cref{lem:unfolding}.
Then
$(1-\rho/3)\mu(\phi(\overline F))\le\widehat\mu\le(1+\rho/3)\mu(\phi(\overline F))$. Since $\mu(\phi(\overline F))=\mu_w(\overline F)$, we have
$(1-\rho/3)\mu_w(\overline F)\le\widehat\mu\le(1+\rho/3)\mu_w(\overline F)$.
Thus the value $\tau\lambda\widehat\mu$ satisfies
$(1-\rho/3)(1-\tau)\mu_w(F)\le \tau\lambda\widehat\mu\le (1+\rho/3)\mu_w(F)$.
Since $\tau=\rho/3$, $\tau\lambda\widehat\mu$ is a
$(1\pm\rho)$-approximation of $\mu_w(F)$, satisfying the approximation
requirement of \Cref{thm:value-reduction}.

Because $\tau=\Theta(\eps/h)$, the largest integer weight is
$O(R/\tau)=O(h/\eps^2)=O(K)$.  The unfolded graph has at most
$O(nK)$ vertices and $O(mK)$ edges, and one weighted-edge update produces
$O(K)$ unweighted updates.  The outer value reduction touches $O(h)$
windows.  Multiplying these costs proves the claimed update-time bound.
The initialization cost is the sum of the initialization
cost of each local window, where every weighted edge contributes to
$\otilde(\eps^{-2})$ unweighted edges in total.
\end{proof}

\subsection{Application\headingperiod}\label{sec:value-Applications}

Now we give an application of our reduction.

\begin{lemma}[\cite{BhattacharyaKS23Size}]
\label{lem:bks-size}
For every constant $\eps>0$, there is a randomized fully dynamic
algorithm that maintains a $(1\pm\eps)$-approximation to the MCM cardinality
of a
graph with at most $n$ nodes and $m$ edges
in $m^{1/2-\Omega_\eps(1)}$ worst-case update time.
The guarantee holds with high probability against an adaptive adversary.
\end{lemma}

\begin{corollary}\label{cor:bks-value}
For every constant $\eps>0$, there is a randomized fully dynamic algorithm that
maintains a $(1\pm\eps)$-approximation to the MWM weight
of a bipartite graph with at most $n$ nodes, $m$ edges and $\poly(n)$ aspect ratio,
in worst-case update time $m^{1/2-\Omega_\eps(1)}$.
The guarantee holds with high probability
against an adaptive adversary.
\end{corollary}

\begin{proof}
Apply \Cref{thm:value-w2u} with the algorithm of
\Cref{lem:bks-size}. For fixed $\eps$, all overhead and parameter
changes in \Cref{thm:value-w2u} are constant, so its update-time bound
simplifies to $m^{1/2-\Omega_\eps(1)}$.
The remaining guarantees are preserved by the reduction.
\end{proof}

\section{Dynamic Weight Reduction for Explicit Matchings\headingperiod}
\label{sec:matching-reduction}

We now show how matching composition can be used as a dynamic weight
reduction framework if one explicitly maintains a
matching. The recourse of a dynamic matching algorithm is the
size of the symmetric difference between the consecutive output
matchings.

We first use a subroutine from \cite{BernsteinCDLST25} to
reduce the aspect ratio to $O(\eps^{-5})$.

\begin{theorem}[{\cite[Theorem 5.9]{BernsteinCDLST25}}]\label{thm:bcdlst:eps-5-reduction}
  Let $\mathcal G$ be a graph family closed under taking subgraphs. Suppose that
  for every $0< \delta\leq 1/2$, there is a dynamic algorithm that, on a graph
  $F\in\mathcal G$ with at most $n$ vertices, $m$ edges and aspect ratio $R$,
  initializes in time $\I(n, m, R, \delta)$, and explicitly maintains a $(1-\delta)$-MWM
  of $F$ in update time $\U(n, m, R, \delta)$ and recourse $\sigma(n,m,R,\delta)$.
  
  Then for every $0<\eps\leq 1/2$, there is a dynamic algorithm that,
  on a graph $F\in\mathcal G$ with at most $n$ vertices, $m$ edges and aspect ratio $\poly(n)$,
  explicitly maintains a $(1-\eps)$-MWM with initialization time
  \[O(\I(n, m, \Theta(\eps^{-5}), \Theta(\varepsilon))+m\varepsilon^{-1}),\]
  update time
  \[O(\U(n, m, \Theta(\eps^{-5}), \Theta(\varepsilon))+\sigma(n,m,\Theta(\varepsilon^{-5}),\Theta(\varepsilon))\eps^{-1}),\]
  and recourse \[O(\sigma(n,m,\Theta(\varepsilon^{-5}),\Theta(\varepsilon))\eps^{-1}).\]
  The reduction preserves worst-case update time and recourse, determinism, and incremental or decremental updates.
\end{theorem}

To complete the reduction, it remains to handle a single
$O(\eps^{-5})$-aspect-ratio instance using only $O(\eps^{-1})$-aspect-ratio
subroutines.  We present the algorithm in
\Cref{sec:lazy-composition-algorithm} and analyze it in
\Cref{sec:lazy-composition-analysis}.  For comparison, using
the implementation of \cite{BernsteinCDLST25} unchanged
would give an additive
$\otilde(\eps^{-6})$ update overhead.  Our lazy-composition algorithm
improves this overhead to $\otilde(\eps^{-3})$.
After the analysis, we give the full reduction and its applications.

\subsection{Lazy Composition for Bounded Aspect Ratio\headingperiod}\label{sec:lazy-composition-algorithm}

Now we consider the case when the aspect ratio is $R=O(\eps^{-5})$.
After rescaling, we may assume the edge weights lie in
$[1,R]$. Let $h=\lceil\log(\eps^{-1})\rceil$ and
$\rho=\Theta(\eps/h)$.
We will use the underlying
algorithm to maintain a $(1-\Theta(\rho/h))$-MWM of
$G_{[i-h,i]}$ for all $0\le i\le L\defeq\lceil\log R\rceil$.
It is clear that every nonempty $G_{[i-h,i]}$ has aspect ratio
$O(\eps^{-1})$. Since we require no recourse guarantee from the underlying
algorithm, to stabilize the local outputs before composing them, we
use the following transformation.

\begin{lemma}[Derived from {\cite[Theorem 5.22]{BernsteinCDLST25}}]\label{thm:bcdlst:low-recourse}
  Let $\mathcal G$ be a graph family closed under taking subgraphs. Suppose that
  for every $0< \delta\leq 1/2$, there is a dynamic algorithm that, on a graph
  $F\in\mathcal G$ with at most $n$ vertices, $m$ edges and aspect ratio $R$,
  initializes in time $\I(n, m, R, \delta)$, and explicitly maintains a $(1-\delta)$-MWM
  of $F$ in update time $\U(n, m, R, \delta)$.

  Then there is a dynamic algorithm that,
  on a graph $F\in\mathcal G$ with at most $n$ vertices, $m$ edges and aspect ratio $O(\eps^{-1})$,
  explicitly maintains a $(1-\rho)$-MWM with initialization time
  \[\I(n, m, O(\eps^{-1}), \Theta(\rho/h)),\]
  amortized update time
  \[\U(n,m,O(\eps^{-1}), \Theta(\rho/h))+\otilde(\eps^{-1}),\]
  and amortized recourse $\otilde(\eps^{-1})$.
  The transformation preserves determinism and incremental or decremental updates.
\end{lemma}

After applying \Cref{thm:bcdlst:low-recourse} to the underlying algorithm,
the local matchings have low recourse. We denote the low-recourse matching
on $G_{[i-h,i]}$ by $M_i$.
Suggested by \Cref{lem:prefix-extension},
for any integer $i\ge0$, we aim to maintain a matching
$P_i\subseteq G_{\le i}$ whose weight is close to
$\mu_w(P_{i-1}\cup M_i)$. Since $P_{i-1}\cup M_i$ has maximum degree
at most two, we borrow another subroutine from \cite{BernsteinCDLST25} that
maintains a $(1-\rho)$-MWM efficiently on degree-two graphs.

\begin{lemma}[{\cite[Algorithm~2]{BernsteinCDLST25}}]
\label{lem:bcdlst-degree-two}
For every $0<\delta\leq 1/2$, there is a deterministic algorithm that,
on a graph $F$ with at most $m$ edges and maximum degree at most two,
fully dynamically maintains an explicit $(1-\delta)$-MWM $M$ with
initialization time $O(m\delta^{-1})$,
worst-case update time $O(\delta^{-1})$ and worst-case recourse $O(\delta^{-1})$.
It also maintains an edge set $E^\prime\subseteq F$ such that:
\begin{enumerate}[label=(\roman*)]
\item $M$ is an exact MWM of $F\setminus E^\prime$;
\item for every connected component $C$ of $F$, $\sum_{e\in E^\prime\cap C}w(e)\le \delta\cdot \sum_{e\in C} w(e)$.
\end{enumerate}
\end{lemma}

We will run \Cref{lem:bcdlst-degree-two} on every $P_{i-1}\cup M_i$. However,
we do not set its output to be $P_i$. Instead, we denote the output
matching by $N_i$ and allow $P_i$ to lag behind $N_i$. We will reset $P_i\gets N_i$ when
$w(N_i)-w(P_i)>\rho w(M_i)$.
The lazy update of $P_i$ ensures the approximation of the matching $P_i$
and prevents the recourse from cascading.
Now we give the full algorithm.

\begin{algorithm}[H]
\caption{Lazy Matching Composition}
\label{alg:bounded-aspect-composition}
\begin{algorithmic}[1]
\State $P_{-1}\gets\varnothing$
\For{$i=0,\ldots,L$}
  \State Apply \Cref{thm:bcdlst:low-recourse} to the underlying dynamic matching algorithm on $G_{[i-h,i]}$
  \State Denote the resulting algorithm by 
  $\A_i$ and the maintained matching by 
  $M_i$
  \State Initialize a copy $\mathcal D_i$ of \Cref{lem:bcdlst-degree-two} on
  $P_{i-1}\cup M_i$ with parameter $\rho$
  \State Let $N_i$ be its maintained matching and set $P_i\gets N_i$
\EndFor
\Procedure{Update}{$e$}
  \State Update every affected $\mathcal A_i$, obtaining the changes
  to $M_i$
  \For{$i=0,\ldots,L$}
    \State Update $\mathcal D_i$ with changes in $P_{i-1}\cup M_i$
    \State Let $N_i$ be the matching maintained by $\mathcal D_i$
    \If{$e$ was deleted from the underlying graph}
      \State $P_i\gets P_i\setminus\{e\}$
    \EndIf
    \If{$w(N_i)-w(P_i)>\rho w(M_i)$}
      \State $P_i\gets N_i$
    \EndIf
  \EndFor
  \State \Return $P_L$
\EndProcedure
\end{algorithmic}
\end{algorithm}

\begin{remark}
To execute the line $P_i\gets N_i$,
we operate the edges in $N_i\oplus P_i$ in time $O(|N_i\oplus P_i|)$.
\end{remark}

\subsection{Analysis of Lazy Composition\headingperiod}\label{sec:lazy-composition-analysis}

We start by analyzing the approximation ratio and recourse of $P_L$.
\paragraph{Approximation ratio.}
We first show that $P_L$ preserves a $(1-O(\eps))$-approximation. The proof is a
direct application of \Cref{lem:prefix-extension}.

\begin{lemma}
\label{lem:lazy-approximation}
After initialization and after every edge update, for every
$0\le i\le L$, we have
\[
 w(P_i)\ge
 \bigl(1-O(i\rho+2^{-h})\bigr)
 \mu_w(G_{\le i}).
\]
In particular, $P_L$ is a $(1-O(\eps))$-MWM of $G$.
\end{lemma}

\begin{proof}
Consider any integer $0\le i\le L$.
It is easy to see that $P_i$ is a matching in $G_{\le i}$.
By \Cref{lem:bcdlst-degree-two}, $w(N_i)\ge(1-\rho)\mu_w(P_{i-1}\cup M_i)$.
The update rule of $P_i$ ensures that
\[
 w(P_i)
 \ge w(N_i)-\rho w(M_i)
 \ge \mu_w(P_{i-1}\cup M_i)-2\rho\mu_w(G_{\leq i}).
\]
Combining this with \Cref{lem:prefix-extension} gives
\[
 w(P_i)-w(P_{i-1})\ge \mu_w(G_{\le i})-\mu_w(G_{\le i-1})-2^{2-h}\mu_w(G_i)-3\rho\mu_w(G_{\leq i}).
\]
Summing over $0\le j\le i$ and applying
\Cref{lem:interval-sum} with $k=0$ gives
\[
 w(P_i)\ge
 \bigl(1-3(i+1)\rho-20\cdot2^{-h}\bigr)\mu_w(G_{\le i}).
\]
Since $h=\lceil\log(\eps^{-1})\rceil$, $L=O(h)$ and $\rho=\Theta(\eps/h)$,
we finish the proof.
\end{proof}

\paragraph{Recourse.}
We let $\Gamma(M)$ denote the total number of edge changes in the support of
a matching $M$, including those forming its initial output.
Let $m_0$ be the number of initial edges and $T$ the number of
subsequent updates.
We assume that fully dynamic and incremental graphs start empty,
so $m_0=0$, while a decremental graph starts with $m_0$ edges and deletes all
of them, so $T=m_0$. Now we bound the amortized recourse of $P_L$,
the output of \Cref{alg:bounded-aspect-composition}.

\begin{lemma}
\label{lem:lazy-recourse}
For any integer $0\le i\le L$, $\Gamma(P_i)=\otilde(T\eps^{-2})$.
In particular, $P_L$ has amortized recourse
$\otilde(\eps^{-2})$.
\end{lemma}

\begin{proof}
At initialization and whenever $P_i$ is reset,
$|N_i\oplus P_{i-1}|\le3|M_i|$: the components containing
an $M_i$-edge have at most $3|M_i|$ edges in total, and all other
nonempty components are isolated $P_{i-1}$-edges retained by $N_i$
by \Cref{lem:bcdlst-degree-two}.

Regard initialization as the first reset. Just after a reset at level
$i$, we have $P_i=N_i$, so $|P_i\oplus P_{i-1}|\le3|M_i|$.
Until the next reset, $P_i$ changes only through forced deletions.
When a deleted edge $e$ is removed from $P_i$, it is already absent
from $P_{i-1}$, since $P_{-1}=\varnothing$ and the levels are processed
in increasing order.
Thus removing $e$ decreases $|P_i\oplus P_{i-1}|$ by one.
Each change in $P_{i-1}$ increases this distance by at most one.
Thus, just before the next reset, $|P_i\oplus P_{i-1}|$ is at most
three times the size of $M_i$ at the previous reset, plus the number
of changes in $P_{i-1}$ since then.
At this point, after any forced deletion, the reset cost satisfies
\[
 |P_i\oplus N_i|\le |P_i\oplus P_{i-1}|+|P_{i-1}\oplus N_i|\le |P_i\oplus P_{i-1}|+3|M_i|.
\]
Therefore, the reset costs at most three times the sum of the sizes
of $M_i$ at the previous and current resets, plus the number of
changes in $P_{i-1}$ between them.
Let $S_i$ be the sum of $|M_i|$ over the noninitial resets at level $i$.
Summing over consecutive reset intervals counts each change in
$P_{i-1}$ at most once. Each noninitial reset appears at most twice:
as the end of one interval and the start of the next.
Its contributions therefore sum to at most $6S_i$.
The initial endpoint and initialization cost $O(m_0)$, and all
forced deletions cost at most $T$. Hence
\[
 \Gamma(P_i)\le\Gamma(P_{i-1})+6S_i+O(m_0+T).
\]

To bound $\sum_i S_i$, sum the gains
$2^{-i}(w(N_i)-w(P_i))$ over all noninitial resets.
The matching weights telescope: this sum is at most the final
$\sum_i2^{-i}w(P_i)$ plus the corresponding weighted losses from
forced deletions. An edge $e\in G_j$ contributes at most
$w(e)\sum_{i\ge j}2^{-i}=O(1)$ across all levels, both when present
at the end and when deleted. Thus the sum of gains is $O(m_0+T)$.
Each reset contributes more than
$\rho2^{-i}w(M_i)\ge\rho2^{-h}|M_i|$, giving
\[
 \sum_i S_i=O\!\left((m_0+T)2^h/\rho\right)
 =\otilde(T\eps^{-2}).
\]
Iterating the recourse inequality over the $L+1=O(h)$ levels proves
the lemma, using $P_{-1}=\varnothing$ and $m_0\le T$.
\end{proof}

Now we are ready to give the full analysis of the lazy composition algorithm.

\begin{lemma}[Bounded-aspect composition]
\label{lem:bounded-aspect-composition}
Let $\mathcal G$ be a graph family closed under taking subgraphs.
Suppose that, for every $0<\delta\le1/2$, there is a dynamic algorithm
that, on a graph $F\in\mathcal G$ with at most $n$ vertices and $m$ edges
and aspect ratio at most $R$, initializes in time
$\I(n,m,R,\delta)$ and explicitly maintains a $(1-\delta)$-MWM in update
time $\U(n,m,R,\delta)$.

Then, for every $0<\eps\le1/2$, there is a
dynamic algorithm that, on a graph $F\in\mathcal G$ with at most $n$
vertices and $m$ edges and aspect ratio $O(\eps^{-5})$,
explicitly maintains a $(1-\eps)$-MWM in initialization time
\[\otilde\left(\I(n,\otilde(m),O(\eps^{-1}),\thetatilde(\eps))+m\eps^{-1}\right),\]
amortized update time
\[\otilde\left(\U(n,m,O(\eps^{-1}),\thetatilde(\eps))+\eps^{-3}\right),\]
and amortized recourse $\otilde(\eps^{-2})$.
The transformation preserves determinism and incremental or decremental
updates.
\end{lemma}

\begin{proof}
Run \Cref{alg:bounded-aspect-composition}, then
\Cref{lem:lazy-approximation} gives approximation $1-O(\eps)$, and
\Cref{lem:lazy-recourse} gives recourse $\otilde(\eps^{-2})$.  The
underlying algorithms are invoked with aspect ratio $O(\eps^{-1})$ and
accuracy $\Theta(\rho/h)=\thetatilde(\eps)$.

There are $L+1=O(h)$ local algorithms $\mathcal A_i$ and
the total initialization time is at most
$\otilde(\I(n,\otilde(m),O(\eps^{-1}),\thetatilde(\eps)))$.
The initialization of all $\mathcal D_i$ inserts $O(Lm)$ edges
at a cost of $O(\rho^{-1})$ time per edge. Thus, the total initialization time is
$O(Lm/\rho)=\otilde(m\eps^{-1})$.

Every edge update affects $\otilde(1)$ local matchings. Thus, it
takes $\otilde\left(\U(n,m,O(\eps^{-1}),\thetatilde(\eps))+\eps^{-1}\right)$
to maintain $M_i$.
By \Cref{lem:lazy-recourse} and $L=\otilde(1)$,
the amortized recourse of all $P_i$ together is $\otilde(\eps^{-2})$.
By \Cref{thm:bcdlst:low-recourse}, the amortized recourse of all
$M_i$ together is $\otilde(\eps^{-1})$, since each update affects at most
$h+1$ local graphs. Forming their initial matchings costs
$O((h+1)m_0)=\otilde(T)$ changes.
Therefore, all $\mathcal D_i$ together receive
$\otilde(\eps^{-2})$ edge updates and require
$\otilde(\eps^{-3})$ update time.
\end{proof}

\subsection{Putting Everything Together\headingperiod}
We now give the complete dynamic weight reduction in general graphs that reduces the 
aspect ratio from $\poly(n)$ to $O(\eps^{-1})$.
\begin{theorem}
\label{thm:dynamic}
Let $\mathcal G$ be a graph family closed under taking subgraphs.
Suppose that, for every $0<\delta\le1/2$, there is a dynamic algorithm that,
on a graph $F\in\mathcal G$ with at most $n$ vertices and $m$ edges and
aspect ratio at most $R$, initializes in time $\I(n,m,R,\delta)$ and
explicitly maintains a $(1-\delta)$-MWM of $F$ in update time
$\U(n,m,R,\delta)$.  For every $0<\eps\le1/2$, there is a dynamic algorithm
that, on a $\poly(n)$-aspect-ratio graph $F\in\mathcal G$ with at most $n$
vertices and $m$ edges, initializes in time
\[
 \otilde\left(\I(n,\otilde(m),O(\eps^{-1}),\thetatilde(\eps))+m\eps^{-1}\right)
\]
and explicitly maintains a $(1-\eps)$-MWM in amortized update time
\[
 \otilde\!\left(
 \U(n,m,O(\eps^{-1}),\thetatilde(\eps))+\eps^{-3}
 \right).
\]
The reduction preserves determinism, and incremental or decremental updates.
\end{theorem}

\begin{proof}
Apply \Cref{lem:bounded-aspect-composition} to the underlying algorithm
to get the transformed $O(\eps^{-5})$-aspect-ratio algorithm. Then apply
\Cref{thm:bcdlst:eps-5-reduction} on it. The runtime follows directly.
\end{proof}

\subsection{Dynamic Weighted-to-Unweighted Reduction
in Bipartite Graphs\headingperiod}

Similar to \Cref{sec:weight-to-cardinality}, we reduce weighted matchings to
unweighted matchings in the unfolded graphs. Its dynamic implementation
reuses a subroutine in \cite{BernsteinCDLST25}.

\begin{lemma}[{\cite[Lemma~5.30]{BernsteinCDLST25}}]
\label{lem:dynamic-unfolding}
Let $F$ be a bipartite graph with integer edge weights between $1$ and $R$.
Suppose an algorithm initializes in time $\I(n,m,\eps)$ and
explicitly maintains a $(1-\eps)$-MCM in update time $\U(n,m,\eps)$.
Then one can initialize an explicit $(1-O(\eps))$-MWM of $F$ in time
\[
 O\!\left(
 \I(O(nR),O(mR),\eps)+m\eps^{-1}\log(\eps^{-1})
 \right)
\]
and maintain it in amortized update time
\[
 O\!\left(
 R\cdot \U(O(nR),O(mR),\Theta(\eps))
 +R\eps^{-2}\log(\eps^{-1})
 \right).
\]
The reduction preserves determinism, and incremental or decremental updates.
\end{lemma}

\begin{theorem}
\label{thm:w2u}
Suppose that, for every $0<\delta\le1/2$, there is a dynamic algorithm
that, on an unweighted bipartite graph with at most $n$ vertices and $m$
edges, initializes in time $\I(n,m,\delta)$ and explicitly maintains a
$(1-\delta)$-MCM in update time $\U(n,m,\delta)$.

Then for every $0<\eps\le1/2$, there is a dynamic algorithm on a bipartite
graph with at most $n$ vertices and $m$ edges and $\poly(n)$ aspect ratio
that initializes in time
\[\otilde\left(\I(\otilde(n\eps^{-2}),\otilde(m\eps^{-2}),\thetatilde(\eps))+m\eps^{-1}\right)\]
and explicitly maintains a $(1-\eps)$-MWM in amortized update time
\[\otilde\left(\eps^{-2}\U(\otilde(n\eps^{-2}),\otilde(m\eps^{-2}),\thetatilde(\eps))+\eps^{-4}\right).\]
The reduction preserves determinism, and incremental or decremental updates.
\end{theorem}

\begin{proof}
Apply the scaling and rounding argument from the proof of
\Cref{thm:value-w2u} to each bounded-aspect instance supplied by
\Cref{thm:dynamic}, using its local accuracy
$\delta=\thetatilde(\eps)$.  It produces positive integer weights in
$[K]$, where $K=\thetatilde(\eps^{-2})$, while losing only an
$O(\delta)$ fraction of every matching's weight.

The proof of \Cref{thm:value-w2u} uses unfolding only to recover the MWM
weight.  Here, instead, apply \Cref{lem:dynamic-unfolding} with $R=K$ and
accuracy $\delta'=\Theta(\delta)$.  The refolded output is an explicit
$(1-O(\delta))$-MWM of the original local instance, and adjusting constants
makes it a $(1-\delta)$-MWM.  Thus \Cref{thm:dynamic} applies.  Substituting
$K=\thetatilde(\eps^{-2})$ and
$\delta'=\thetatilde(\eps)$ into the initialization and update bounds
of \Cref{lem:dynamic-unfolding}, and then into \Cref{thm:dynamic}, gives
the two bounds.  Both reductions preserve determinism and partial
dynamism.
\end{proof}

\subsection{Applications\headingperiod}\label{sec:explicit-applications}
Now we discuss the implications of our reduction frameworks for obtaining 
new dynamic $(1-\eps)$-MWM algorithms. For those results whose $\eps$ dependence
was stated with $\poly(\eps^{-1})$ in \cite{BernsteinCDLST25}, we do not
state our improvements in the hidden $\eps$ dependence.
Below we emphasize four applications.
\begin{center}
\small
\begin{tabular}{@{}p{0.19\linewidth}p{0.285\linewidth}
                    p{0.25\linewidth}p{0.155\linewidth}@{}}
\toprule
Setting & Prior Result & Our Result & Reduction\\
\midrule
\shortstack[l]{Incremental\\Bipartite}
& \shortstack[l]{$\otilde(n\eps^{-9}\log n+m\eps^{-8})$\\
  \cite{BernsteinCDLST25}\\
  (see correction below)}
& \shortstack[l]{$\otilde(n\eps^{-8}\log n+m\eps^{-7})$\\
  \Cref{cor:incremental-bipartite}}
& \shortstack[l]{\Cref{thm:w2u}}\\
\midrule
\shortstack[l]{Fully Dynamic\\General, Degree $\Delta$}
& \shortstack[l]{$\otilde(\Delta\eps^{-5})$\\
  {\cite{BernsteinCDLST25}}}
& \shortstack[l]{$\otilde(\Delta\eps^{-3})$\\
  \Cref{cor:low-degree}}
& \shortstack[l]{\Cref{thm:dynamic}}\\
\midrule
\shortstack[l]{Fully Dynamic\\General, Arboricity $\alpha$}
& \shortstack[l]{No prior $\alpha\poly(\eps^{-1})$\\bound is known}
& \shortstack[l]{$\otilde(\alpha\eps^{-4})$\\
  \Cref{thm:lowarb}}
& \shortstack[l]{\Cref{thm:dynamic}}\\
\midrule
\shortstack[l]{Fully Dynamic\\General}
& \shortstack[l]{$\otilde(\sqrt m\,\eps^{-4}+\eps^{-6})$\\
  \cite{BernsteinCDLST25}}
& \shortstack[l]{$\otilde(\sqrt m\,\eps^{-3})$\\
  \Cref{cor:bcdlst-style-fully-dynamic-general}}
& \shortstack[l]{\Cref{thm:dynamic}}\\
\bottomrule
\end{tabular}
\end{center}

\begin{lemma}[\cite{BlikstadK23}]
\label{lem:incremental-mcm}
There exists a deterministic incremental algorithm that maintains an explicit
$(1-\eps)$-MCM in an $n$-vertex $m$-edge bipartite graph
in total time $O(n\eps^{-6}+m\eps^{-5})$.
\end{lemma}

\begin{corollary}
\label{cor:incremental-bipartite}
There exists a deterministic incremental algorithm that maintains an
explicit $(1-\eps)$-MWM in an $n$-vertex, $m$-edge,
$\poly(n)$-aspect-ratio bipartite graph in total time
$\otilde(n\eps^{-8}\log n+m\eps^{-7})$.
\end{corollary}

\begin{proof}
Apply \Cref{lem:incremental-mcm} with accuracy $\thetatilde(\eps)$
to the unfolded instances in \Cref{thm:w2u}.
The fixed weight windows give $\otilde(\log n)$ incremental instances,
each with $\otilde(n\eps^{-2})$ vertices, while each original edge
creates $\otilde(\eps^{-2})$ unfolded edges in total.
Summing the total-time bounds gives
$\otilde(n\eps^{-8}\log n+m\eps^{-7})$, which also covers the
$\otilde(m\eps^{-4})$ composition and refolding overhead.
\end{proof}

Applying the same black-box accounting to the reduction
in~\cite[Lemma~6.2]{BernsteinCDLST25} gives total time
$\otilde(n\eps^{-9}\log n+m\eps^{-8})$.
The $\log n$ factor accounts for summing the vertex-dependent
term over the $O(\log n)$ outer weight windows, which may share
vertices. This factor is absent from the bound stated there.

\begin{lemma}[\cite{GuptaP13}]
\label{lem:bounded-degree-black-box}
There is a deterministic fully dynamic algorithm that explicitly maintains
a $(1-\eps)$-MWM on a graph of maximum degree $\Delta$ and aspect ratio at most $R$
in amortized update time $\otilde(\Delta R\eps^{-2})$.
\end{lemma}

\begin{corollary}
\label{cor:low-degree}
There is a deterministic fully dynamic algorithm that explicitly maintains
a $(1-\eps)$-MWM on a $\poly(n)$-aspect-ratio graph of maximum degree $\Delta$
in amortized update time $\otilde(\Delta\eps^{-3})$.
\end{corollary}

\begin{proof}
Apply \Cref{thm:dynamic} to \Cref{lem:bounded-degree-black-box}.
\end{proof}

\begin{theorem}
\label{thm:lowarb}
There is a deterministic fully dynamic algorithm that explicitly
maintains a $(1-\eps)$-MWM in amortized update time
$\otilde(\alpha\eps^{-4})$ if the graph has arboricity at most $\alpha$ after every
update. 
\end{theorem}

\begin{proof}
Apply \Cref{thm:dynamic} to \Cref{lem:bounded-degree-black-box}
and maintain the low-degree sparsifier in \Cref{lem:dynamic-lowarb-sparsifier}
on $\otilde(\eps^{-1})$-aspect-ratio
instances.
\end{proof}

\begin{lemma}[\cite{GuptaP13}]
\label{lem:general-graph-black-box}
There is a deterministic fully dynamic
algorithm that explicitly maintains a $(1-\eps)$-MWM on a graph with at
most $m$ edges and aspect ratio at most $R$ in update time
$\otilde(\sqrt m R\eps^{-2})$.
\end{lemma}

\begin{corollary}
\label{cor:bcdlst-style-fully-dynamic-general}
There is a deterministic fully dynamic
algorithm that explicitly maintains a $(1-\eps)$-MWM on a
$\poly(n)$-aspect-ratio graph with at most $m$ edges in amortized update time
$\otilde(\sqrt m\eps^{-3})$.
\end{corollary}

\begin{proof}
Apply \Cref{thm:dynamic} to \Cref{lem:general-graph-black-box}.
\end{proof}
Let $\M_G\defeq\operatorname{conv}
  \{\mathbf 1_M : M \text{ is a matching in } G\}$ denote the matching polytope of $G$.

\begin{lemma}[\cite{BernsteinCDLST25}]
\label{lemma:sparsify-to-low-degree}
Let $\bx\in\M_G$ be initially zero and undergo entry updates
that preserve $\bx\in\mathcal M_G$.  For every $\eps>0$, there is a
deterministic algorithm that explicitly maintains a subgraph $H\subseteq\supp(\bx)$
in $\otilde(\eps^{-1} \polylog n)$ amortized time per entry update to $\bx$,
such that $H$ has maximum degree $\otilde(\eps^{-2}\polylog n)$ and
$\mu_w(H)\ge (1-O(\eps))\bw^\top\bx$.
\end{lemma}

\begin{corollary}
Let $\bx\in\M_G$ be initially zero and undergo entry updates
that preserve $\bx\in\mathcal M_G$.  There is a deterministic dynamic
rounding algorithm that maintains a matching
$M\subseteq\supp(\bx)$
satisfying $w(M)\ge (1-\eps)\bw^\top\bx$
with $\otilde(\eps^{-6} \polylog n)$ amortized update time.
\end{corollary}

\begin{proof}
    Apply \Cref{cor:low-degree} to \Cref{lemma:sparsify-to-low-degree}.
\end{proof}

\section*{AI Disclosure\headingperiod}
OpenAI's ChatGPT 5.6 Sol model was used for interactive research discussions
guided by the author's research notes and proposed directions.
The model suggested an initial version of
\Cref{alg:bounded-aspect-composition}, improving upon the original dynamic
implementation of \cite{BernsteinCDLST25}.
It also provided an alternative analysis of an author's initial
$O(\alpha \eps^{-3})$-degree $(1-\eps)$-MWM sparsifier for graphs with arboricity at most $\alpha$, sharpening the degree bound to
$O(\alpha\eps^{-1})$ in \Cref{app:lowarb-proof}.
OpenAI Codex was used for polishing and proofreading.
The author independently verified all mathematical claims and takes full
responsibility for the paper.

\section*{Acknowledgments\headingperiod}
Thank you to the anonymous reviewers for their helpful feedback and discussions. The author was supported by Dantzig-Lieberman Operations-Research Fellowship.

\bibliographystyle{alpha}
\bibliography{reference}
\appendix
\section{Low-Degree MWM Sparsifier for Low-Arboricity Graphs\headingperiod}
\label{app:lowarb-proof}
\cite{Solomon18} constructs a $(1-\eps)$-MCM sparsifier for graphs
of arboricity at most $\alpha$ by having each vertex mark
$O(\alpha/\eps)$ incident edges and retaining only those marked by both
endpoints.  We give a weighted analogue in which each vertex marks its
$O(\alpha/\eps)$ heaviest incident edges.
\begin{lemma}
\label{lem:weighted-matching-sparsifier}
Let $H$ contain the edges marked by both endpoints when
each vertex marks its $d$ heaviest incident edges, so $\Delta(H)\leq d$.
If $G$ has arboricity
at most $\alpha$ and $d>2\alpha$,
then $\mu_w(H)\ge\left(1-\frac{2\alpha}{d}\right)\mu_w(G)$.
\end{lemma}

\begin{proof}
By definition of arboricity,
$E(G)$ can be partitioned into $\alpha$ forests.
In every tree component,
choose an arbitrary root and orient all edges toward it.  Each vertex then
has at most one outgoing edge per forest, and hence out-degree at most
$\alpha$.  If a vertex $v$ does not mark an incident edge, then among the
$d$ edges it marks, at least $d-\alpha$ are oriented into $v$. Denote this
set by $K_v$.

Fix an MWM $M$, write $M_0=M\cap H$ and $B=M\setminus H$,
and choose for each $e\in B$ an endpoint that does not mark $e$.  Place one
unit of mass at that endpoint.  Because $B$ is a matching, no vertex receives
more than one initial unit.

Process the mass in rounds.  At a vertex $v$, split all its mass equally
among the edges of $K_v$.  A part sent along $xv\in K_v$ is absorbed on
$xv$ if $x$ also marks $xv$; otherwise it moves to $x$ for the next round.
The new edge is at least as heavy as the edge that brought the mass to $v$
(or the originating edge in the first round), because $v$ marks the former
but not the latter.  Thus weights never decrease along a route.

Let $m_t$ be the maximum mass processed at any vertex in round $t$, so
$m_0\le1$.  Mass transferred into a vertex $x$ arrives along an edge
oriented out of $x$.  There are at most $\alpha$ such edges, and each carries
at most $m_t/(d-\alpha)$ mass.  Therefore $m_{t+1}\le m_t\cdot \alpha/(d-\alpha)$
and every routed edge carries in total at most
$\frac1{d-\alpha}\sum_{t\ge0}m_t\le\frac1{d-2\alpha}$ mass.
Denoting by $x_e$ the mass absorbed on $e\in H$, we then have
$\sum_{e\in H}w(e)x_e\ge w(B)$.

Set $y=\mathbf 1_{M_0}+x$.  Besides the initial mass at $v$, all
absorbed mass incident to $v$ can be charged to routed mass on an edge
oriented out of $v$.  There are at most $\alpha$ such edges, each carrying
at most $1/(d-2\alpha)$ mass in total.  Moreover, an initial unit at $v$ and an edge
of $M_0$ incident to $v$ cannot both exist, because $M$ is a matching.
Consequently, $\sum_{e\in E_v} y_e\le 1+\frac{\alpha}{d-2\alpha}\le \frac{d}{d-2\alpha}$.
For every odd $U\subseteq V$, since each edge carries at most $\frac{1}{d-2\alpha}$ mass
and the arboricity is at most $\alpha$,
$\sum_{e\in E(U)} y_e\le\frac{|U|-1}{2}+\frac{|E_G(U)|}{d-2\alpha}\le\frac{d}{d-2\alpha}\frac{|U|-1}{2}$.
Thus $(d-2\alpha)/d\cdot y$ satisfies all degree and odd-set constraints of the
matching polytope and is thus a valid fractional matching in Edmonds' matching polytope \cite{Edmonds65}.
Therefore, $\mu_w(H)\ge\frac{d-2\alpha}{d}w(y)\ge\left(1-\frac{2\alpha}{d}\right)\mu_w(G)$.
The degree bound is immediate.
\end{proof}

\begin{lemma}
\label{lem:dynamic-lowarb-sparsifier}
For $0<\eps\le1/2$, there is a deterministic fully dynamic
algorithm that, on a graph $F$ with arboricity at most $\alpha$ and aspect ratio
at most $R$, maintains a subgraph $H$ satisfying $\Delta(H)=O(\alpha/\eps)$ and
$\mu_w(H)\ge(1-\eps/2)\mu_w(F)$ with
$O(\eps^{-1}\log R)$ worst-case update time, and causes $O(1)$
edge changes to $H$.
\end{lemma}

\begin{proof}
Round weights down to powers of $1+\eps/4$ and set
$d=\lceil8\alpha/\eps\rceil$.  Each vertex marks its
$\min\{d,\deg(v)\}$ heaviest incident edges under the rounded weights,
breaking ties deterministically, and $H$ contains the edges marked by both
endpoints.  By \Cref{lem:weighted-matching-sparsifier},
$\Delta(H)\le d$ and $\mu_w(H)\ge\frac{1-2\alpha/d}{1+\eps/4}\mu_w(F)\ge(1-\eps/2)\mu_w(F)$.

The rounded weights have $O(\eps^{-1}\log R)$ distinct values.
At each vertex, store the marked and unmarked edges in lists by
value.  After an update, scan all the distinct values
and restore the invariant by at most one mark swap. 
This takes
$O(\eps^{-1}\log R)$ worst-case time, changes $O(1)$ marks, and hence causes $O(1)$ edge changes in $H$.
\end{proof}

\end{document}